\documentclass[]{article}

\usepackage{graphicx}
\usepackage{amsmath}
\usepackage{natbib}
\usepackage{amsfonts}
\usepackage{amssymb}  
\usepackage{amsthm}  
\usepackage{subfigure}
\usepackage{verbatim}
\usepackage{color}
\usepackage{algorithm}
\usepackage{algorithmic}
\usepackage{pdflscape}

\renewcommand{\d}{\text{d}}

\newcommand{\N}{\mathbb{N}}
\newcommand{\R}{\mathbb{R}}
\newcommand{\I}{1{\hskip -2.5 pt}\hbox{I} }

\newcommand{\dSN}[4]{f_{SN}(#1; #2,#3,#4)}

\newtheorem{theorem}{Theorem}
\newtheorem{corollary}{Corollary}
\newtheorem{lemma}{Lemma}

\title{Bayesian nonparametric location-scale-shape mixtures}
\author{Antonio Canale\thanks{Department of Economics and Statistics, University of Turin and Collegio Carlo Alberto, Italy  \mbox{(\emph{antonio.canale@unito.it})}}$\,$  and Bruno Scarpa\thanks{Department of Statistical Sciences, University of Padua, Italy  \mbox{(\emph{scarpa@stat.unipd.it})}}  }

\begin{document}

\maketitle 

\begin{abstract}
Discrete mixture models are one of the most successful approaches for density estimation. Under a Bayesian nonparametric framework,  Dirichlet process location-scale mixture of Gaussian kernels is the golden standard, both having nice theoretical properties and computational tractability. In this paper we explore the use of the skew-normal kernel, which can naturally accommodate several degrees of skewness by the use of a third parameter. The choice of this kernel function allows us to formulate nonparametric location-scale-shape mixture prior with large support and good performance in different applications. 
Asymptotically, we show that this modelling framework is consistent in frequentist sense. Efficient Gibbs sampling algorithms are also discussed and the performance of the methods are tested through simulations and applications to galaxy velocity and fertility data. Extensions to accommodate discrete data are also discussed.
\end{abstract}

{\center \textbf{Keywords: }}
Dirichlet process; large support; posterior consistency; rounded mixture priors; skew-normal distribution

\section{Introduction}

Discrete mixture models are routinely used for univariate and multivariate density estimation. A discrete mixture model characterizes the density of $y \in \mathcal{Y} \subset \R$ as 
\begin{equation}
	f(y) = \sum_{h=1}^k \pi_h K(y;\theta_h)
\label{eq:mix1}
\end{equation}
where $\sum_{h=1}^k \pi_h = 1$ and $K(\cdot;\theta)$ is a kernel function parametrized by a vector of parameters $\theta$.
In~(\ref{eq:mix1}), $k$ can be any finite integer leading to a finite mixture model, or $\infty$ leading to an infinite, or nonparametric, mixture model. 
Bayesian mixture models generalize model (\ref{eq:mix1}) by
\begin{align}
& f(y)  =  \int K(y; \theta) dP(\theta), \,\,\, P \sim \Pi, \notag
\end{align}
where $P$ is a mixing measure (in equation \eqref{eq:mix1} this measure is discrete), and $\Pi$ is a prior over the space of mixing measures. Stick breaking priors \citep{ishw:jame:2001} are convenient choices for $P$ since a draw from a stick-breaking prior is a discrete probability measure almost surely. Among them the most used is the Dirichlet process (DP) prior \citep{art:ferg:1973,art:ferg:1974}. A Dirichlet process mixture (DPM) model can be written in form \eqref{eq:mix1} marginalizing out $P$, namely  
\begin{equation}
	f(y) = \sum_{h=1}^\infty \pi_h K(y;\theta_h), \qquad \theta_h \stackrel{iid}{\sim} P_0, \qquad \pi = \{\pi_h\} \sim \mbox{Stick}(\alpha)
\end{equation}
where $P_0$ is a base probability measure and 
Stick$(\alpha)$ denotes the stick-breaking process by \citet{seth:1994} with positive scalar parameter $\alpha$.
The choice of a Gaussian kernel $K(\cdot;\theta)$
gives the DPM of Gaussians \citep{art:lo:1984, art:esco:west:1995} which is computationally convenient and has nice theoretical properties. For example, it has been proved \citep{art:lo:1984}
that it can approximate any continuous density, including asymmetric, fat tailed and  multimodal ones; e.g., kernels with similar locations but different scales can lead to heavy-tailed and skewed distributions.  In addition, from the Bayesian asymptotic point of view the DPM of Gaussians prior leads to posterior consistency \citep{ghos:etal:1999,barr:etal:1999,tokd:2006,cana:debl:2013}, so if  $f_0$ is the true density that generates the data, under mild regularity conditions, the posterior concentrates on a $\epsilon$-neighborhood of $f_0$ with a given rate contraction \citep{ghos:etal:2000,ghos:vand:2001,ghos:vand:2007,walk:etal:2007,shen:etal:2012}.

If $y \in \mathcal{Y} \subset \N$ is a discrete random variable, the same ideas can be ideally extended into the settings of  probability mass function estimation but, in this case,  limited literature is available.  A common strategy is to use a mixture of Poissons or negative binomials, which unfortunately are quite restrictive. \citet{cana:duns:2011} recently proposed to induce a discrete kernel by rounding a continuous kernel. 
A rounded mixture of Gaussians prior has been showed to be successful to fit simulated and real data, and to inherit the strong theoretical support from the Gaussian case. 

An interesting feature of finite mixture models, both for continuous and count observations, is the induced clustering structure \citep{fral:raft:2002},  so that each component can be seen as a cluster of units whose results are usually clearly interpretable.
However, a common concern  is related to the number of mixture component allowed, i.e., it may happen that redundant mixture components with similar locations and scales are estimated. Clearly this form of overfitting may lead to an unnecessarily complex model which is particularly unappealing if the sample size is small, and it induces a lack of interpretability due to the overlapping of similar kernels.
To deal with this problem \citet{petr:etal:2013} propose a repulsive mixture prior which favors well separated components and  can lead to more interpretable clustering structure. 

Clearly, when the data actually show different sub-populations, the choice of Gaussian kernel  leads to symmetric clusters.
However, if these sub-populations are not symmetric, this procedure can fail to detect the real sub-population structure.
For example, by considering the data about the global cognition scores of 451 patients suffering from Alzheimer's disease, \citet{art:fruh:pyne:2010} show that
DPM of Gaussians estimates quite well the density of the global cognition scores, using 3 mixture components. However, the data are clearly bimodal and the mixture of Gaussians needs 3 mixture components only to fit the skewness of the data.

To deal with this issue, a mixture of more flexible kernels, which accounts for several degrees of skewness, may be appropriate to obtain well-separated asymmetric clusters. 
To this end, we explore the use of the \citet{art:azza:1985}'s skew-normal kernel within the nonparametric mixture model framework which allows the model to retain both computational tractability and good theoretical properties. 
The case of discrete variables can also be included in this framework by exploiting the rounding procedure of \citet{cana:duns:2011} with skew-normal kernels in place of classic Gaussian kernels.
 
Finite mixtures of skew-normals have been already discussed in the literature both in the frequentist and Bayesian context. \citet{art:lin:etal:2007} discuss a finite mixture of skew-normal model assuming the number of components to be fixed. These authors propose
 an EM and a Gibbs Sampling algorithm for the frequentist and Bayesian estimation of the parameters, respectively. \citet{art:fruh:pyne:2010}, in fully Bayesian setting, discuss mixtures of skew-normal and skew-$t$, motivated by multivariate data arising from biotechnological applications. They provide an interesting discussion about the number of components, involving reversible jump MCMC and evaluation of posterior probability via information criteria. However, from a practical point of view, it is not clear how to choose the number $k$ of components, and in practice they fixed it \textsl{a priori}.
\citet{losc:etal:2013} propose a DPM of skew-normal to estimate densities, obtaining promising results on some simulation scenarios.
However, to our knowledge, no theoretical properties have been proved so far for nonparametric skew-normal mixture models. In this paper, we will analyze the properties of location-scale-shape mixture models using the skew-normal kernel by showing large support of the prior and proving strong posterior consistency. We also introduce a new model for probability mass function estimation based on rounded skew-normal kernels.
In addition, we propose efficient sampling algorithms, which exploit recent advances in Bayesian inference for the skew-normal model \citep{cana:scar:2013}, both for the continuous and count cases.

The rest of the paper is organized as follows. Section 2 reviews the skew-normal distribution and formalizes location-scale-shape mixture models. Section 3 discusses the asymptotic properties of the DPM of skew-normals prior. Section 4 gives the posterior full conditional distributions representation from which a Gibbs sampling algorithm can be obtained. In Section 5 some simulation studies are carried out to show the performance of the methods in finite samples. Section 6 provides two applications and Section 7 concludes the paper. 

\section{Models}

\subsection{The skew-normal distribution}

A random variable $X$ is distributed as a skew-normal \citep{art:azza:1985} with location $\xi$, scale $\omega$ and shape $\lambda$, written $X \sim SN(\xi, \omega,\lambda)$, if its density function  is
\begin{equation}
\label{eq:sn}
	f_{SN}(X;\xi, \omega,\lambda) = \frac{2}{\omega} \phi\left(\frac{x-\xi}{\omega}\right)  \Phi\left( \lambda \frac{x-\xi}{\omega} \right),
\end{equation}
where $\phi(x)$ and $\Phi(\cdot)$ are the density function and the distribution function, respectively, of a standard normal,
$\xi \in \R$, $\omega \in \R^+$ and $\lambda \in \R$. Note that for $\lambda = 0$ the density reduces to the normal $N(x;\xi,\omega^2)$. Let $F_{SN}(x;\xi, \omega,\lambda)$ be the correspondent cumulative distribution function.

The skew-normal model has several stochastic representations \citep[see, e.g.,][]{azza:2014}. Some of them are interesting since they mimic real life phenomena, and others are convenient because of their nice mathematical construction. An elegant and useful stochastic representation, 
for example, is obtained via convolution. If $Z\sim  N(0,1)$ and $V \sim N(0,1)$, and $\delta \in (-1, 1)$, then 
\begin{equation}\label{eq:stocastica}
	X = \delta |Z| + \sqrt{1-\delta^2} V
\end{equation}
has a skew-normal distribution $X \sim SN(0, 1,\delta/\sqrt{1-\delta^2})$. The latter representation is particularly useful if we want to simulate skew-normal random variable and, after suitable adaptation, it will be used in the Gibbs sampling algorithm of Section~\ref{sec:gibbs}.

\subsection{Mixtures of skew-normals}
\label{sec:dpmsn}

Assume $y$ a continuous random variable, $y \sim f$ and $f \in \cal L$ where $\cal L$ is the space of densities with respect to the Lebesgue measure. 
A prior on $\cal L$, is a DPM of skew-normal  if
\begin{equation}
 f(y)  =  \sum_{h=1}^\infty \pi_h f_{SN}(y; \xi_h, \omega_h, \lambda_h)
\label{eq:dpmsn}
\end{equation}
with $\pi \sim \mbox{Stick}(\alpha)$, and $(\xi_h, \omega_h, \lambda_h ) \stackrel{iid}{\sim} P_0$. To conclude the prior specification,  we may assign gamma hyperprior to $\alpha$ as suggested by  \citet{art:esco:west:1995}. Namely $\alpha \sim \text{Ga}(a_\alpha, b_\alpha)$ with $a_\alpha$ and $b_\alpha$ small in order to have a distribution with heavy tails and favoring smaller values of $\alpha$. To denote a general density from the mixture model \eqref{eq:dpmsn} we use the notation $f_{MSN}$.

The choice of $P_0$ is very important both from the applied and theoretical point of view. $P_0$ is a measure over $\R \times \R^+ \times \R$ and needs to be specified.
 In mixture of Gaussians models the usual choice for $P_0$ is normal-inverse-gamma for gaining conjugacy in the blocked Gibbs samplers.  In specifying $P_0$ here, we want to retain computational tractability while having the possibility to include, if present, prior information. A recent proposal for the Bayesian analysis of the skew-normal model has been discussed  by \citet{cana:scar:2013}, 
showing that the prior 
\begin{equation}
	P_0(\xi,\omega,\lambda) = N(\xi; \xi_0, \kappa \omega^2) \times \text{Ga}(\omega^{-2};a, b) \times N(\lambda;0,\psi_0),
	\label{eq:basemeasure}
\end{equation}
leads to closed form full conditional posterior distributions whose sampling can be efficiently carried out within a Gibbs sampling scheme. See Section \ref{sec:gibbs} for further details.
Note that the marginal prior for $\lambda$ is a normal centered in zero with variance $\psi_0$. This implies that the prior expected skewness for each mixture component is zero.
However, if we are motivated  by finding clustering patterns and we expect that  most cluster has positive (negative) skewness, the marginal prior for $\lambda$ can be generalized and assumed to be skew-normal with suitable parameters \citep{cana:scar:2013}.

\subsection{Mixture of rounded skew-normals}
\label{sec:dpmrsn}

Consider, now, the case in which $y \in \N$  to be a discrete or count random variable with $y \sim p$ and $p\in \cal C$ where $\cal C$ is the space of the probability mass functions on the integers.
Following \citet{cana:duns:2011}, assume that $y = h(y^*)$, where $h(\cdot)$ is a rounding function defined so that $h(y^*) = j$ if $y^* \in (a_j,a_{j+1}]$, for $j=0,1,\ldots, \infty$, with $a_0 < a_1 < \ldots $ an infinite sequence of pre-specified thresholds that defines a disjoint partition of $\R$ with $a_0 = -\infty$ and $a_\infty = \infty$. Under this setting the probability mass function $p$ of $y$ is $p = g(f)$, where $g(\cdot)$ is the rounding function having the simple form 
\begin{equation}
\label{eq:mapping.g}
	p(j) = g(f)[j] = \int_{a_j}^{a_{j+1}} f(y^*) \d y^* \, \, \, \, \, \,  j \in \N.
\end{equation}
A prior over $\cal C$ is obtained specifying a prior for the distribution of the latent $y^*$. Our proposal consists in 
\begin{equation}
 y = h(y^*), \quad y^* \sim f^*, \quad
 f^*(y)  =  \sum_{h=1}^\infty \pi_h f_{SN}(y^*; \xi_h, \omega_h, \lambda_h)
\label{eq:rsn}
\end{equation}
with $\pi \sim \mbox{Stick}(\alpha)$, and $(\xi_h, \omega_h, \lambda_h ) \sim P_0$ as in Section~\ref{sec:dpmsn}. We call this formulation DPM of rounded skew-normal.

Clearly, the properties of the prior induced on the space of probability mass functions, here described, will be largely driven by the properties of the prior on the latent space. 
In the next section we will study first some of the properties of model \eqref{eq:dpmsn} and then discuss the discrete case.

\section{Large support and posterior consistency}
\label{sec:properties}

An important property  that a Bayesian nonparametric procedure should hold is the  consistency in frequentist sense
of the final posterior,
 namely 
  if a fixed density $f_0$ has generated the data, the posterior should concentrates on a small neighborhood of such $f_0$ as the sample size increases.


We first concentrate on the asymptotic properties of model \eqref{eq:dpmsn}. 
Large support of the prior is an important property while also having a crucial role in posterior consistency.  The Kullback-Leibler (KL) support of the prior $\Pi$ is the set of all $f_0$ such that $\Pi(\mathcal{K}_\epsilon(f_0) ) > 0$, where $\mathcal{K}_\epsilon(f_0)$ is a KL $\epsilon$-neighborhood of $f_0$. 
\citet{art:wu:ghos:2008} proved the prior positivity of Kullback-Leibler  $\epsilon$-neighborhoods
 under mild regularity conditions on $f_0$, for DP location-scale mixture of several kernels. Among them, the authors considered the skew-normal kernel too, assuming the shape parameter as fixed. Under the theory therein for each fixed $\lambda_0$ we have that the prior on the space of continuous univariate densities induced via 
\[
f(y; P, \lambda_0)  =  \int f_{SN}(y; \xi, \omega, \lambda_0) dG(\xi, \omega ), \,\,\,\, G \sim DP(\alpha G_0) 
\]
has large KL support. 

The next theorem, which instead is in terms of location-scale-shape mixtures prior formalizes the size of the KL support of prior \eqref{eq:dpmsn}. The proof is reported in the Appendix.
\begin{theorem}
\label{theo:support}
Let $f_0$ be a density over $\R$ with respect to Lebesgue measure and let $\Pi$ denote the prior on $f$ induced from a location-scale-shape mixture of skew-normal kernels, i.e.
\begin{equation}
f(x; P)  =  \int f_{SN}(x; \xi, \omega, \lambda) dP(\xi, \omega, \lambda ), \,\,\,\, P \sim \tilde{\Pi}.
\label{eq:lssmixture}
\end{equation}
Assume that the weak support of $\tilde{\Pi}$ contains all probability measures on $\R \times \R^{+} \times \R$ that are compactly supported and that: 
	$(i)$ $0 < f_0(x) < M$ for some finite constant $M$, 
	$(ii)$ $|\int f_{0}(x) \log f_0 (x) \d x| < \infty$, 
	($iii$) for some $a>0$, $\int f_0(x) \log \frac{f_0(x)}{\psi_a(x)} \d x < \infty$, where $\psi_a(x) = \inf_{ t \in (x-a, x+a)} f_0(t)$, and ($iv$) for some $\eta >0$, $\int |x|^{2(1+\eta)}f_0(x) \d x < \infty$. Then $f_0$ is in the KL support of $\tilde{\Pi}$.
\end{theorem}

The conditions on $f_0$ required by Theorem~\ref{theo:support} are the same conditions for the KL support of general location-scale mixtures and can be seen as standard regularity and tail conditions. 
As a corollary of Theorem~1, we give the following result which formalizes the size of the support of the prior \eqref{eq:rsn}.
The proof follows directly from Theorem 1 of \citet{cana:duns:2011} and hence is omitted.

\begin{corollary}
Let $p_0$ be a probability mass function on $\N$ such that $p_0 \in g(\mathcal{L}_{\Pi^*})$ where $g$ is the mapping function in \eqref{eq:mapping.g}, $\Pi^*$ a prior defined as in \eqref{eq:lssmixture} and $\mathcal{L}_{\Pi^*}$ is the KL support of $\Pi^*$. Say $\Pi$ the prior induced by $\Pi^*$ as described in Section~\ref{sec:dpmrsn}, then $p$ is in the KL support of $\Pi$.
\label{cor:count}
\end{corollary}

Weak posterior consistency is a direct consequence of the KL condition on the prior thanks to the theory of \citet{art:schw:1965}. This means that as the sample size increases the posterior probability of any weak neighborhood around the true data-generating distribution $f_0$ converges to one with $P_{f_0}$-probability~1. 
However, strong posterior consistency is more interesting.
Weak consistency implies strong consistency in the discrete probability mass function case \citep[see Theorem~2 of][]{cana:duns:2011} and hence, for the mixtures discussed in \ref{sec:dpmrsn},  Corollary~\ref{cor:count} is sufficient for strong posterior consistency too. 

To prove strong consistency for the mixture \eqref{eq:dpmsn}, we need some further conditions on the prior. Let first $J(\delta, \mathcal{L})$ denote the $L_1$ metric entropy of the set $\cal L$, defined as the logarithm of $N(\delta,{\cal L})$,  the minimum integer $N$ for which there exists $f_1,\ldots,f_N\in {\cal L}$ such that ${\cal L}\subset\bigcup_{j=1}^N\{f:\, ||f - f_j||_1 <\delta\}$.
To obtain strong posterior consistency we need to define a sieve, i.e., a sequence of sets which eventually grows to cover the whole parameter space satisfying the requirements of Theorem 8 of \citet{ghos:etal:1999}. That Theorem basically requires that such a sieve has low entropy and high prior mass. To construct our sieve we exploit the stick-breaking representation of the Dirichlet process following an approach first proposed by \citet{pati:duns:tokd:2011} and adapting it to the more challenging case of skew-normal kernels. To build our sieve we first introduce the set
\begin{equation}
	\mathcal{F}_{a,u,l,s,m} = \left\{ f_{MSN}:|\xi_h|< a, l < \omega_h < u, |\lambda_h|< s, \mbox{for $h =1, \dots, m$}, \sum_{h>m} \pi_h<\epsilon  \right\}
\label{eq:sieve}
\end{equation}
and, in the following, we formalizes its size in terms of metric entropy $J(\delta,  \mathcal{F}_{a,u,l,s,m})$. 

\begin{lemma}
For some $a>0$, $u>l>0$, and $s>0$, the set $\mathcal{F}_{a,u,l,s,m}$ of \eqref{eq:sieve} has
\[
	J(\epsilon,\mathcal{F}_{a,u,l,s,m}) \leq m \log\left\{ d_1 \left(\frac{as}{l}\right) + d_2 \left(\frac{a}{l}\right) + d_3 s \log\left(\frac{u}{l}\right) + d_4 \log\left(\frac{u}{l}\right)  + s +1\right\} + d_3 m \log(d_4 m)	
\]
where $d_1$, $d_2$, $d_3$, and $d_4$ are constants depending on $\epsilon$.
\label{theo:entropy}
\end{lemma}

To conclude this section we give our main result on consistency for the model \eqref{eq:dpmsn} with base measure \eqref{eq:basemeasure} which combines Theorem 8 of \citet{ghos:etal:1999} and Lemma~\ref{theo:entropy}.

\begin{theorem}
Assume we observe an iid sample $y = (y_1 , \dots , y_n)$ from $f_0$
satisfying the conditions of Theorem~\ref{theo:support}. For any $\epsilon >0$, if $\Pi$ is the the prior defined by \eqref{eq:dpmsn}--\eqref{eq:basemeasure}, then the posterior $\Pi(\{f:||f - f_0||_1 <\epsilon\}\mid y_1, \dots, y_n) \to 1$  a.s. $P_{f_0}$.
\label{theo:final}
\end{theorem}

\begin{proof}
First define the set $\mathcal{F}_{n}$ as the set in \eqref{eq:sieve} with $a = O(\sqrt{n})$, $s = O(\sqrt{n})$, $l = O(1/\sqrt{n})$, $u=O(\exp\{n\})$, and $m = O(n/\log(n))$. Then, the proof relies on showing that  $\mathcal{F}_{n}$ satisfies the conditions of Theorem 8 of \citet{ghos:etal:1999}. This is obvious from the definition of $P_0$ in \eqref{eq:basemeasure} and our Lemma \ref{theo:entropy}. 
\end{proof}

\section{Computation}
\label{sec:gibbs}

A Gibbs sampler for the mixture of skew-normals can be developed generalizing the blocked Gibbs sampler of \citet{ishw:jame:2001}. We introduce latent $S_1, \dots, S_n$ where $S_i = h$ if the $i$-th subject is drawn from the $h$-th mixture component. With such an approach, conditionally on $S_i$, each observation is drawn from a single skew-normal distribution and hence the updated of each cluster-specific set of parameters can be done easily. To this end, using the stochastic representation  (\ref{eq:stocastica}), we also introduce latent half-normal distributed variables $\eta_1, \dots, \eta_n$. Conditionally on those variables the observations can be seen as drawn from a suitable Gaussian distribution and this allows us to gain conjugacy for the location and scale parameters of each component of the mixture. 

Finally, the distributions for the shape parameters are in closed forms and belong to the unified-skew-normal class of distribution \citep[discussed in][with the acronym SUN]{arel:azza:2006} as discussed in \citet{cana:scar:2013}. The precision parameter $\alpha$ can be updated as in  \citet{art:esco:west:1995}. 
The complete Gibbs sampler for model \eqref{eq:dpmsn} is reported in Algorithm~1.

\begin{algorithm}                      
\caption{Gibbs sampling for posterior simulation of model (4)}
\label{alg1}                           
\begin{enumerate}     
	\item Sample $S_i$, the class indicator from the multinomial 
	\[	\text{Pr}(S_i = h | -) = \frac{\pi_h  f_{SN}(y_i| \xi_h, \omega_h, \lambda_h) }{\sum_{l=1}^H \pi_l  f_{SN}(y_i| \xi_l, \omega_l, \lambda_l) }
		\]
		with $h = 1, \dots, H$ and $H$ the number of occupied clusters.
	\item Sample $\alpha$ using \citet{art:esco:west:1995} given $n$ and $H$, the number of occupied clusters
	\item Update the stick-breaking weights using
	\[
		 V_h \sim \text{Be}\left(1 + n_h,  \alpha + \sum_{l=h+1}^H n_l \right)
	\]
	where $n_h$ is the sample size of the $h$th cluster.
	\item Update 
	\[	\eta_i \sim N ( \delta_{S_i} (y_i^* - \xi_{S_i}), \omega_{S_i}^2 (1-\delta_{S_i}^2)) \]
	where $\delta_h$ is $\lambda_h/\sqrt{\lambda_h^2+1}$.
	\item Sample $(\xi_h,\omega_h)$ from
	\[ 
		N\left(\hat{\mu}_h, \hat{\kappa}_h \omega_h^2  \right) \text{InvGam}(a + n_h/2 + 1, b + \hat{b}_h)
	\]
	where
	\begin{align*}
		\hat{\mu}_h & = \frac{\kappa \sum_{S_i = h} ( y_i - \delta_h \eta_i) + (1-\delta_h^2)\xi_0 }
				{n_h + \kappa \omega^2 (1-\delta_h^2) } \\
		\hat{\kappa}_h & = \frac{\kappa(1-\delta_h^2)}{n_h\kappa + (1-\delta_h^2)} \\
		\hat{b}_h & = \frac{1}{2(1-\delta^2_h)} \left\{ \sum_{S_i = h} \eta_i^2 -2\delta_h \sum_{S_i = h} \eta_i(y_i - \xi_h) + \sum_{S_i = h} (y_i - \xi_h)^2  + (1-\delta^2_h)(\xi_h - \xi_0)^2 \right\}.
	\end{align*}
	\item Sample $\lambda_h$ from 
	\begin{equation}
	\label{eq:sun}
		\lambda_h \sim SUN_{1,n_h}(\lambda_h;0,0,\Delta_h,\Gamma_h)
	\end{equation}
	where $\Delta_h = [\delta_i]_{i=1,\dots,n_h}$ with $\delta_i = \psi_0 y_i (\psi_0^2 y_i^2 + 1)^{-1/2}$,  and 	$\Gamma_h  = I - D(\Delta_h)^2 + \Delta_h \Delta_h^T$, where $D(V)$ is a diagonal matrix whose elements coincide with those of the vector $V$.
\end{enumerate}
\end{algorithm}

For posterior computation in the discrete case, an additional data augmentation step and a modification of step 1 are required. Indeed we first need to generate the latent continuous variable $y^*$ and then we can continue on the line of the Gibbs sampler for the continuous case. Algorithm~2 gives the Gibbs sampler for model \eqref{eq:rsn}.

\begin{algorithm}
\caption{Gibbs sampling for posterior simulation of model (4)}
\label{alg2}
\begin{itemize}
	\item[0] For $i = 1, \dots, n$, generate $y_i^*$ from the full conditional posterior	
		\begin{itemize}
		\item[$0a$] Generate $u_i \sim U \Big(F_{SN}( a_{y_i} ; \xi_{S_i}, \omega_{S_i}, \lambda_{S_i} ), F_{SN}( a_{y_i+1} ; \xi_{S_i}, \omega_{S_i}, \lambda_{S_i} )\Big)$ 
		\item[$0b$] Let $y^*_i = F_{SN}^{-1}(u_i; \xi_{S_i}, \omega_{S_i}, \lambda_{S_i} )$ 
		\end{itemize}
	\item[1b] Sample $S_i$, the class indicator from the multinomial 
	\[	\text{Pr}(S_i = h | -) = \frac{\pi_h  p(y_i| \xi_h, \omega_h, \lambda_h) }{\sum_{l=1}^H \pi_l  p(y_i| \xi_l, \omega_l, \lambda_l) }
		\]
		with $h = 1, \dots, H$ and $H$ the number of occupied clusters.
	\item[2b] Continue with the Gibbs sampler for the continuous case (Algorithm 1) with $y_i^*$ in place of $y_i$;
\end{itemize} 
\end{algorithm}


\section{Simulation studies}
\label{sec:simulation}
\subsection{Density estimation}
\label{sec:densityestimation}

To assess the performance of the proposed approach, we conducted a simulation study comparing our location-scale-shape mixture of skew-normal with a classic location-scale mixture of Gaussians. Several simulations have been run under different settings obtaining similar results and, in the following, we will report the results for four scenarios. The first simulation case assumed that the data were simulated as a mixture of three Gaussians, $  0.35 N(-2, 1) + 0.5 N(4,2) + 0.15 N(5,2.5)$, the second scenario, as a mixture of two skew-normal, $0.65 SN(0,1,5) + 0.35 SN(4,2,3)$, the third as a mixture of a Gamma and a Gaussian, $0.25 \text{Ga}(2,1) + 0.75 N(3,1)$, while the last one as a simple exponential distribution with mean parameter 2. 

For each scenario, we generated sample of sizes $n=50,100,200$ and we fit the two mixture models to $1{,}000$ replicated data sets.
The methods were compared based on a Monte Carlo approximation to the mean Kullback-Leibler divergence and $L_2$ distance, defined as
\begin{equation}
\label{eq:kll2}
	KL(f,g) = \int f(x) \log(f(x)/g(x)) \d x, \quad \, L_2(f, g) = \left( \int (f(x) - g(x))^2 \d x\right)^{1/2}.
\end{equation}
In implementing the blocked Gibbs samplers of the two models the first $1{,}000$ iterations were discarded as a burn-in and the next $5{,}000$ samples were used to calculate the posterior mean of the density on a fine grid of points of the domain.  For our mixture of skew-normals we choose, as hyperparameters, $\xi_0 = \overline{y}$, the sample mean, and $\kappa = s^2$, the sample variance, $\psi_0 = 10$, and $a = b = 1/2$. Hyperparameters for the mixture of Gaussian were fixed as: the location mean $\mu_0 = \overline{y}$, the location scale $\kappa = s^2$, and the precision gamma hyperparameters equal to $\nu_1 = \nu_2 = 1$.
For the precision parameter of the DP prior we assigned a Gamma hyperprior as in \citet{art:esco:west:1995} in both cases. The values of the density for a wide variety of points of the domain were monitored to check for convergence and mixing. The results of the simulation are reported in Table~1. 

\begin{table} 	\footnotesize								
\caption{Kullback-Leibler divergence and $L_2$ distance for the mean posterior densities, posterior mean number of occupied cluster components and posterior mean of the DP precision parameter}	
\centering	
\begin{tabular}{llrrrrrrrr} \hline									
&	&	\multicolumn{4}{c}{Scenario 1: {\sl mix of normals}} &	\multicolumn{4}{c}{Scenario 2: {\sl mix of skew-normals}} \\
$n$ &	Kernel 	&	KL &	$L_2$ &	$E(k|-)$ &	$E(\alpha|-)$ &	KL &	$L_2$ &	$E(k|-)$ &	$E(\alpha|-)$ \\
\hline									
50 &	Gaussian 	& 	0.237 &	0.146 &	3.940 &	0.908 &	0.391 &	0.280 &	3.137 &	0.694 \\
&	Skew-normal & 	0.312 &	0.158 &	2.800 &	0.614 &	0.465 &	0.286 &	3.241 &	0.720 \\
100 &	Gaussian 	& 	0.107 &	0.099 &	4.039 &	0.788 &	0.261 &	0.242 &	3.282 &	0.622 \\
&	Skew-normal & 	0.118 &	0.103 &	2.788 &	0.519 &	0.260 &	0.225 &	3.209 &	0.606 \\
200 &	Gaussian 	& 	0.051 &	0.070 &	3.970 &	0.671 &	0.182 &	0.209 &	3.703 &	0.619 \\
&	Skew-normal & 	0.051 &	0.071 &	2.736 &	0.445 &	0.148 &	0.178 &	3.369 &	0.558 \\
&	&	\multicolumn{4}{c}{Scenario 3: {\sl mix gamma+normal}} &	\multicolumn{4}{c}{Scenario 4: {\sl exponential}} \\
$n$ &	Kernel &	KL &	$L_2$ &	$E(k|-)$ &	$E(\alpha|-)$ &	KL &	$L_2$ &	$E(k|-)$ &	$E(\alpha|-)$ \\
\hline									
50 &	Gaussian 	& 	0.260 &	0.165 &	3.482 &	0.789 &	1.237 &	0.497 &	3.884 &	0.884 \\
&	Skew-normal & 	0.308 &	0.163 &	3.331 &	0.749 &	1.284 &	0.495 &	4.161 &	0.955 \\
100 &	Gaussian 	& 	0.140 &	0.124 &	3.989 &	0.775 &	0.967 &	0.455 &	4.617 &	0.905 \\
&	Skew-normal & 	0.159 &	0.125 &	3.556 &	0.681 &	0.959 &	0.445 &	4.664 &	0.917 \\
200 &	Gaussian 	& 	0.073 &	0.089 &	4.405 &	0.748 &	0.781 &	0.414 &	5.374 &	0.927 \\
&	Skew-normal & 	0.078 &	0.090 &	3.791 &	0.635 &	0.720 &	0.392 &	5.364 &	0.926 \\
\hline									
\end{tabular}									
\label{tab:continue}									
\end{table}			

The mixture of Gaussians fit often requires a higher  number of occupied clusters. 
Our location-scale-shape mixture has generally comparable performances in terms of Kullback-Leibler and $L_2$ distance from the truth. However, for small samples ($n=50$), our proposal has higher measures of distance if compared with the mixture of Gaussians. In fact, this is not surprising, since it is well known that often the inference with the skew-normal model is not particularly efficient for small sample sizes. As expected, for high $n$ our method is perfectly comparable and sometimes preferable, to the mixture of Gaussians.

\subsection{Probability mass function estimation}
\label{sec:pmfestimation}

A second Monte Carlo experiment has been conducted to assess the performance of the proposed approach with respect to the rounded mixture of Gaussians of \citet{cana:duns:2011}. 
 Also here, we report the results only for four scenarios, although different simulation settings lead to  similar conclusions. 
The first simulation case assumed the data were simulated from a 3-values probability mass function defined as $p(2) = p(4) = 0.2$, $p(3)=0.6$ and $p(j)=0$ for $j \notin \{2, 3, 4\}$, the second scenario, assumed the data were simulated from an underdispersed probability mass function, the Conway-Maxwell-Poisson distribution \citep{shmu:etal:2005} with parameters $\lambda=3$ and $\nu=5$, the third and the fourth assumed the data were simulated from mixtures of Poissons, namely $0.65 \text{Po}(2.5) + 0.35 \text{S-Po}(0.5, 9)$ and $0.6 \text{Po}(0.5) + 0.4 \text{R-Po}(0.5, 12)$, where S-Po($\lambda$, $\nu$) denotes Poisson distribution shifted to the right of the value $\nu$ and R-Po($\lambda, \gamma$) is a reversed Poisson distribution with negative skewness and probability distribution defined as $p(y=j) = C(\lambda, \gamma) \lambda^{\gamma - j} \exp\{-\lambda\}$, where $ C(\lambda, \gamma)$ is a suitable normalizing constant, $\gamma$ is an integer and $\lambda \in \R^+$. 
For each scenario, we generated $1{,}000$ data sets for each of the following sample size: $n=50,100,200$. As before, the methods were compared based on a Monte Carlo approximation to the mean Kullback-Leibler divergence and $L_2$ distance defined in \eqref{eq:kll2}, where here the integrals are with respect to the counting measure. 

In implementing the blocked Gibbs samplers of the two models the first $1{,}000$ iterations were discarded as a burn-in and the next $5{,}000$  samples were used to calculate the posterior mean of the probability mass function between 0, 1, \dots up to a suitable upper bound definied on a case-by-case basis. We set the hyperparameters similarly to the previous section. The results of the simulation are reported in Table~2. 

\begin{table}	\footnotesize								
\caption{Kullback-Leibler divergence and $L_2$ distance for the mean posterior probability mass functions, posterior mean number of occupied cluster components and posterior mean of the DP precision parameter}	
\begin{tabular}{llrrrrrrrr} \hline									
&	&	\multicolumn{4}{c}{Scenario 5: {\sl 3-values distribution}} &	\multicolumn{4}{c}{Scenario 6: {\sl Compoisson}} \\
$n$ &	kernel &	KL &	$L_2$ &	$E(k|-)$ &	$E(\alpha|-)$ &	KL &	$L_2$ &	$E(k|-)$ &	$E(\alpha|-)$  \\ \hline
50 &	Gaussian 	& 	1.229 &	0.124 &	1.282 &	0.263 &	0.113 &	0.123 &	4.031 &	0.924 \\
&	skew-normal & 	1.051 &	0.081 &	2.143 &	0.465 &	0.175 &	0.153 &	3.123 &	0.691 \\
100 &	Gaussian 	& 	1.244 &	0.105 &	1.222 &	0.216 &	0.080 &	0.098 &	4.512 &	0.886 \\
&	skew-normal & 	1.122 &	0.057 &	1.862 &	0.341 &	0.086 &	0.108 &	3.388 &	0.644 \\
200 &	Gaussian 	& 	1.259 &	0.105 &	1.222 &	0.216 &	0.062 &	0.083 &	5.298 &	0.915 \\
&	skew-normal & 	1.179 &	0.057 &	1.862 &	0.341 &	0.050 &	0.082 &	3.862 &	0.650 \\
&	&	\multicolumn{4}{c}{Scenario 7: {\sl mix of Poissons}} &	\multicolumn{4}{c}{Scenario 8: {\sl mix of R-Poissons}} \\
$n$ &	kernel 	&	KL &	$L_2$ &	$E(k|-)$ &	$E(\alpha|-)$ &	KL &	$L_2$ &	$E(k|-)$ &	$E(\alpha|-)$ \\ \hline 
50 &	Gaussian 	& 	0.116 &	0.148 &	3.458 &	0.775 &	0.069 &	0.236 &	1.489 &	0.31 \\
&	skew-normal & 	0.200 &	0.155 &	1.068 &	0.216 &	0.063 &	0.178 &	2.334 &	0.512 \\
100 &	Gaussian 	& 	0.091 &	0.132 &	4.189 &	0.815 &	0.052 &	0.205 &	1.411 &	0.252 \\
&	skew-normal & 	0.107 &	0.125 &	1.280 &	0.229 &	0.024 &	0.105 &	2.113 &	0.392 \\
200 &	Gaussian	 & 	0.075 &	0.120 &	5.753 &	1.005 &	0.045 &	0.192 &	1.389 &	0.219 \\
&	skew-normal & 	0.075 &	0.107 &	2.545 &	0.422 &	0.012 &	0.075 &	1.690 &	0.271 \\
\hline									
\end{tabular}									
\label{tab:simdiscrete}									
\end{table}

The posterior mean number of occupied clusters in the rounded skew-normal mixture is not always smaller than that of the Gaussian case. Particularly this is not true in Scenario 5 and 6, where we still get a better fit in terms of Kullback-Leibler divergence and $L_2$ distance from the truth. Note that in the latter scenarios, the probability mass function is concentrate in a small number of points of the domain. 

\section{Applications}

\subsection{Galaxy data}
\label{sec:galaxy}

First we applied our modeling framework to the famous Galaxy dataset \citep{roed:1990}. The dataset consists on the velocity of 82 galaxies. The histogram of the speeds reveals that the data are clearly multimodal. This feature supports the Big Bang theory, as the different modes of density can be though as clusters of galaxies moving at different speed. The data analysis was already carried out via DP mixture of Gaussians by \citet{art:esco:west:1995}, and we compare their results with our mixture of skew-normal. 

In implementing our blocked Gibbs sampler the first $1{,}000$ iterations were discarded as a burn-in and the next $10{,}000$  samples were used to calculate the posterior mean of the density on a fine grid of points of the domain. As a default non informative choice, we set the hyperparameters $\xi_0 = \overline{y}$, $\kappa = s^2$, $\psi_0 = 10$, and $a = b = 1/2$.
Since the scientific interest is galactic clustering, we followed  \citet{art:esco:west:1995} in letting the precision DP parameter $\alpha \sim \text{Ga}(1/2, 50)$. The posterior mean predictive density is plotted in Figure~\ref{galaxy-fit} along with the empirical histogram and the estimate obtained via DP mixture of Gaussians. Our fitted density turns out to be smoother, in particular around the central area of the domain where the DPM of Gaussians clearly detects two separate modes. 

The  posterior distribution of the number of occupied clusters in the two models is reported in  Figure~\ref{galaxy-clus}. It is evident that our approach leads to a generally lower number of occupied clusters and that the posterior distribution of the number of clusters is coherent with the number of observed modes of the density. Indeed, if a galactic cluster is skewed, a single skew-normal component is sufficient, while two or more mixture components with collapsing modes are needed when using Gaussian kernels. 

\begin{figure}[ht]
\begin{center}
\includegraphics[scale=.55]{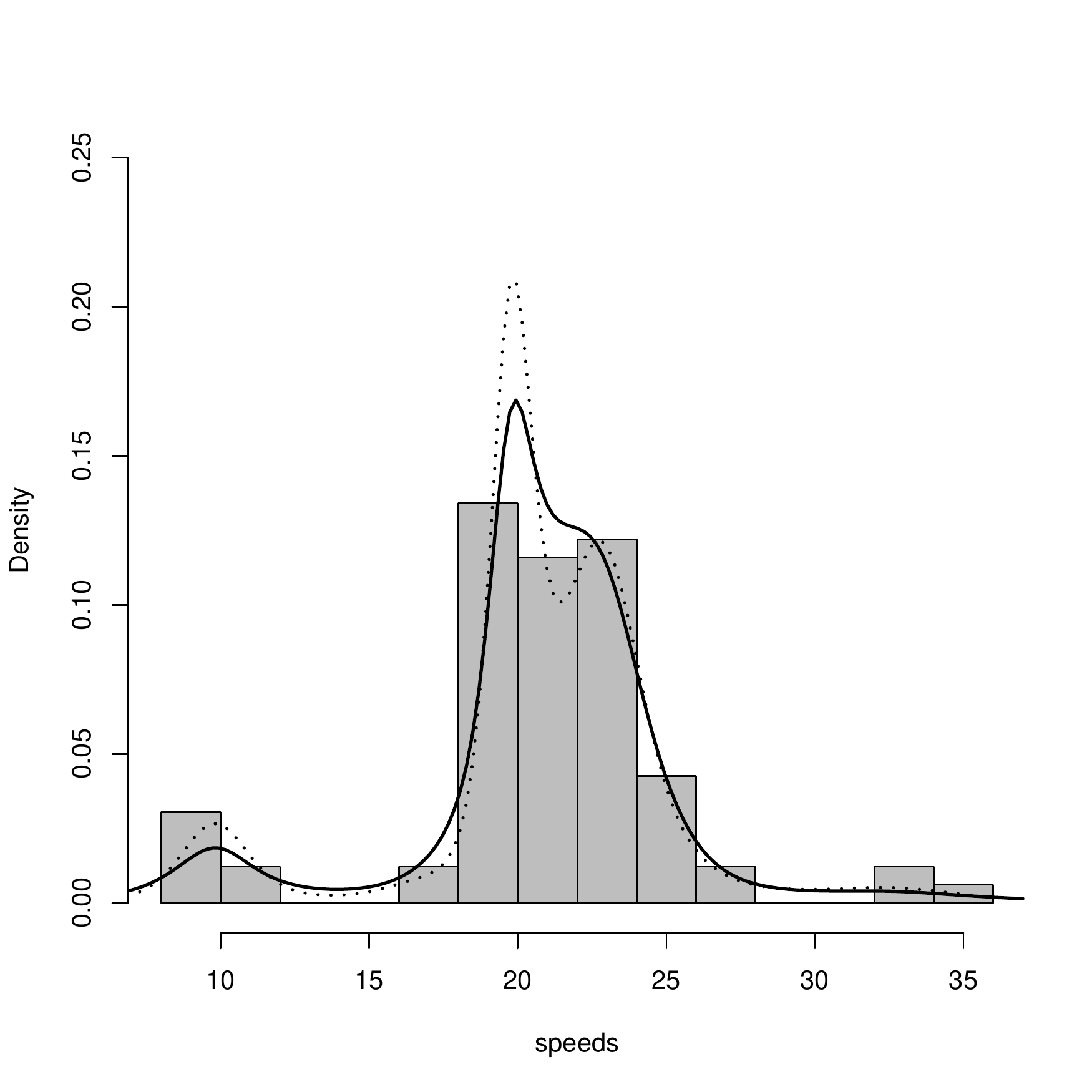}
\end{center}
\caption{Posterior estimated densities for the location-scale-shape mixture of skew-normal (continuous line) and of location-scale mixture of Gaussians (dotted line) along with the histogram of the galaxy data}
\label{galaxy-fit}
\end{figure}

\begin{figure}[ht]
\begin{center}
\includegraphics[scale=.35]{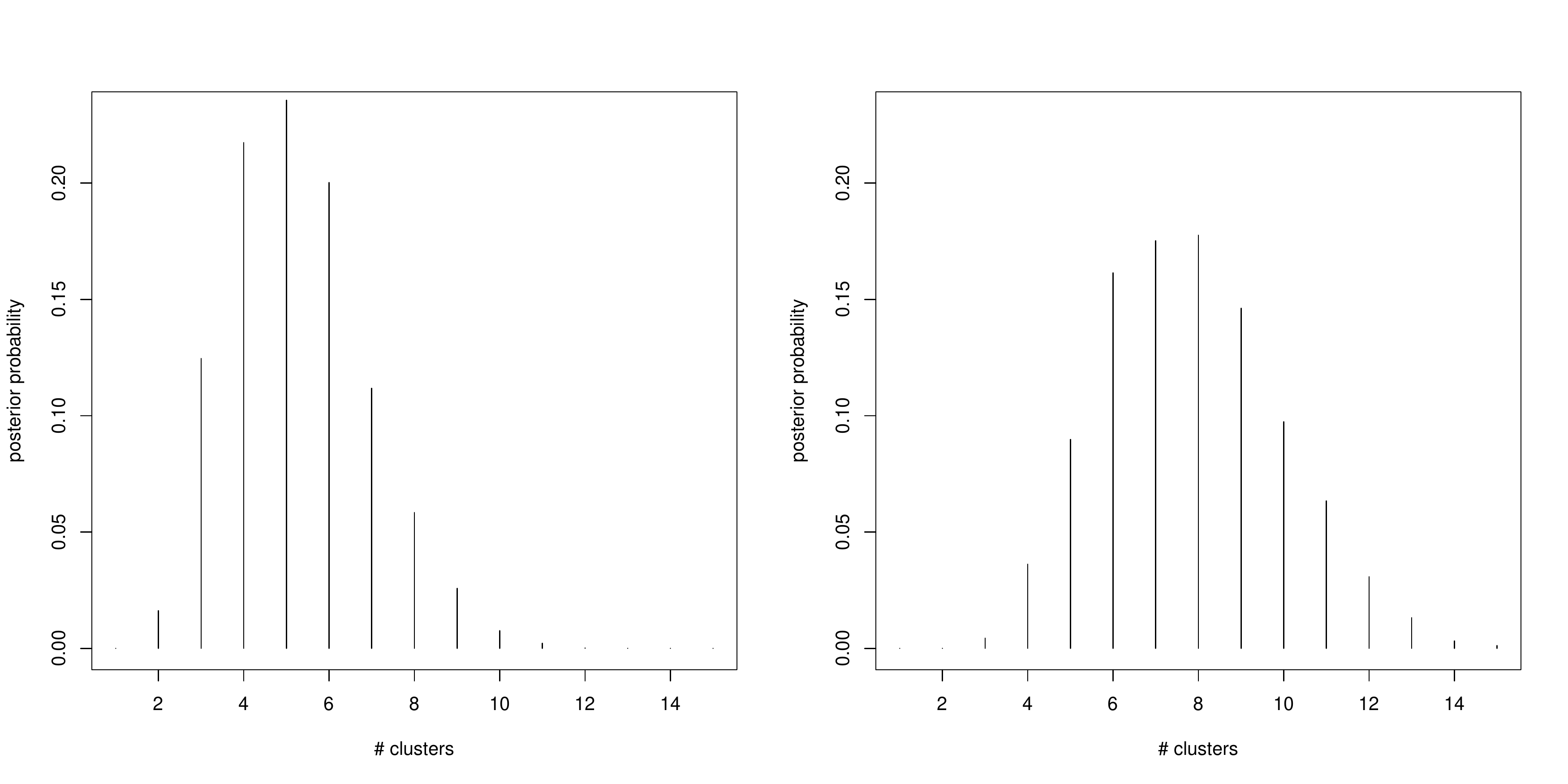}
\end{center}
\caption{Posterior probability of the average number of occupied clusters in the location-scale-shape mixture of skew-normal (left) and of location-scale mixture of Gaussians (right) for the galaxy dataset}
\label{galaxy-clus}
\end{figure}

\subsection{Childbirth age data}
\label{sec:milan}


We apply our modeling framework to data on the births in the Milan municipality in 2011 divided by areas 
to estimate the different age-specific probability of childbirth. 
Milan is one of the biggest and multiethnic cities in Italy being the center of many economic activities and the destination of strong national and international immigration. In this context, fertility may be affected by socio-demographical and economical differences among and within the different urban areas. The presence of different subpopulations with different educational level, socio-economic status or citizenships, inside each area may give rise to asymmetric distributions of the age of the mother at childbirth.
For small populations, such as residents in Milan, there are not many specific studies on fertility indicators, and we may expect different behaviors of women with respect to the age at childbirth. 
Given this variety of possible patterns, a nonparametric approach to density estimation seems appropriate to both smooth the random noise affecting the curves, and to account for different patterns.

The use of mixture models in demography is not new. In the context of country age-specific fertility rate estimation, for example, several finite mixture models have been discussed \citep{Chandola1999,KohlerOrtega2000,PeristeraKostaki2007,Schmertmann2003}.

Let $y$ be the age of the mother at childbirth and assume that we want to model the probability distribution $p(y)$. 
In fact, even if age is ideally continuous, data are rounded to the lower integer. Hence $p(y)$ is a probability mass function defined on the positive integers and thus we estimate $p(\cdot)$ with model \eqref{eq:rsn}. In this context, the use of the skew-normal kernel has the advantage that if the data present different sub-populations whose fertility curves are typically not symmetric, each of them may be fitted by one specific component. Otherwise if a single asymmetric population is present, a single skew-normal may be sufficient to obtain a satisfactory fit. 

To implement our Gibbs sampler, the first $1{,}000$ iterations were discarded as a burn-in and the next $5{,}000$ draws were used to calculate the posterior mean of the probability mass function for $15, \dots, 50$ years of the women.  As posterior estimate, we consider the mean probability mass functions in the nine zones, reported in Figure~\ref{sn-empirica} along with the empirical estimate. 

Our procedure allows for smoothing across the age of childbirth and this is evident in Figure~\ref{sn-empirica}, where the mean of the posterior probability mass function is smoother than the empirical estimate, which has an erratic behavior.  However, our procedure is also able to catch the shape of each probability mass function. For example, zone 1, 3, and 5 are almost symmetric with, in zone 3, only mild left skewness, and in zone 1 high concentration around the mean. These probability mass functions clearly show a delay in childbirth, with respect to classical curves, but also suggest the presence of a common fertility behavior inside these areas. Other areas, instead, present a small hump around 20-25 years. In zone 4 and 6 this is fairly evident, while in zone 8 and 9 this is only partially noticeable. The former areas are likely to have at least two subpopulations, with the smaller consisting in women anticipating the childbirth. Most of the estimated probability mass functions exhibit moderate skewness to the left, sign of a general trend of the majority of women in the area to postpone the age at childbirth, but also indicator of the presence of subgroups that anticipate it.

\begin{landscape}
\begin{figure}[ht]
\begin{center}
\includegraphics[scale=.65]{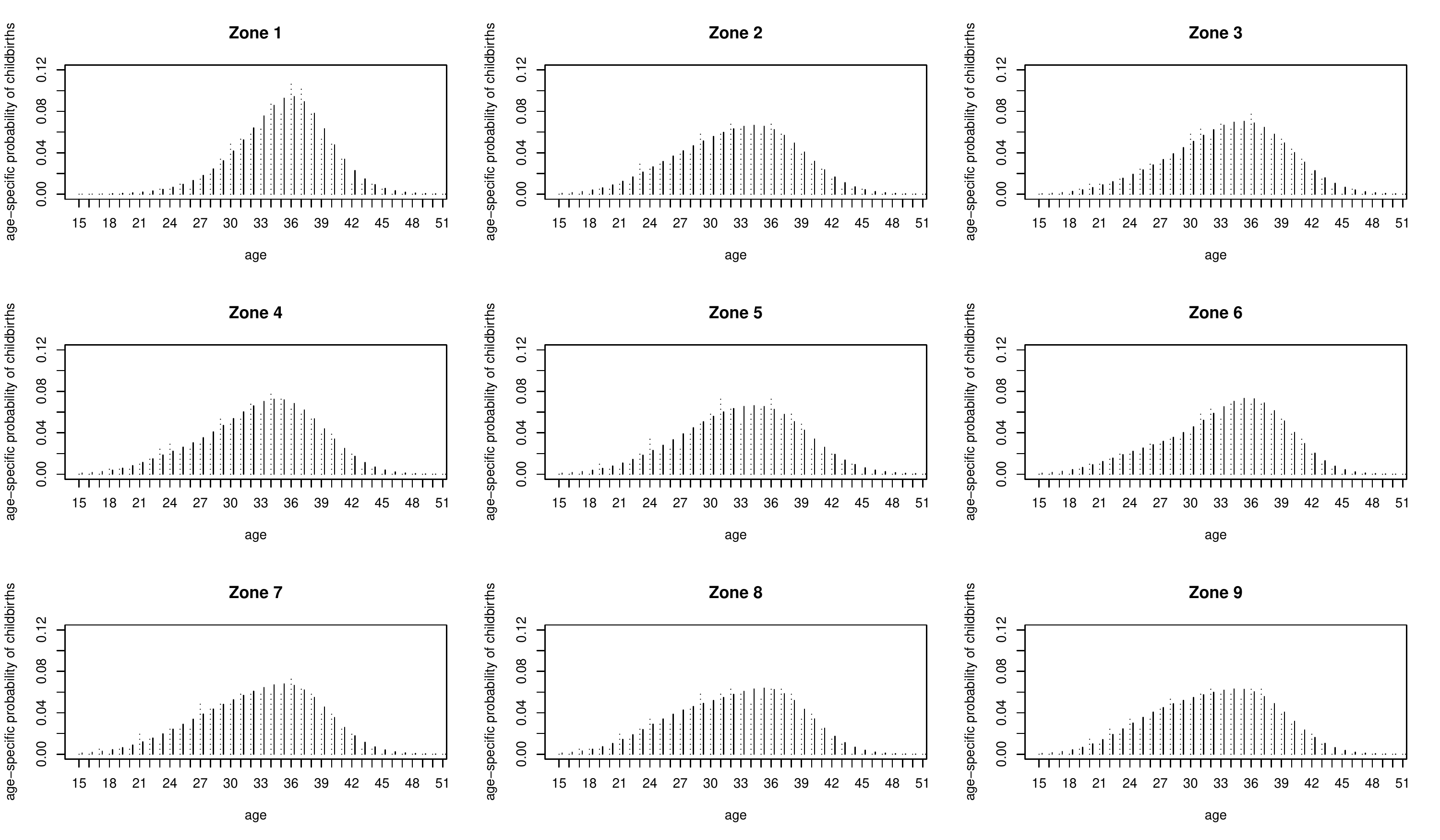}
\end{center}
\caption{Posterior mean probability mass function (black) and empirical probability mass function (dotted) for the age of the mother at childbirth in the nine zone of Milan.}
\label{sn-empirica}
\end{figure}
\end{landscape}


\section{Discussion}

In this paper we have discussed nonparametric location-scale-shape mixture of skew-normal kernels for density estimation and its extension to model discrete probability mass functions.
These classes of models have the particular advantage of determining clusters with different shapes, allowing for several degrees of positive and negative skewness. This has been shown to have an impact in real applications where the model-based clustering may have some specific interpretation. 
We showed that this class of models has large support and asymptotic posterior consistency. 
Simulations confirm the asymptotic behavior showing a substantial equivalence in the quality of fitting between the mixture of Gaussians and the mixture of skew-normals. However the number of occupied clusters is typically quite smaller in our model, thus allowing easier interpretation, when it is needed. 
Future research can  investigate multivariate extensions of location-scale-shape mixture models.  


\section*{Appendix}

\begin{proof}[Theorem 1]
The proof is similar to that of Theorem 3.2 of \citet{tokd:2006} with major adaptations involving the properties of the skew-normal distribution. We report all the passages for sake of completeness. Note that
\begin{equation}
\int f_0(x) \log \frac{f_0(x)}{f(x)} \d x = \int f_0(x) \log \frac{f_0(x)}{\tilde{f}(x)}  + \int f_0(x) \log \frac{\tilde{f}(x)}{f(x)},
\label{eq:kl}
\end{equation}
where $\tilde{f}$ and $f$ are densities obtained via a mixture prior of the type in \eqref{eq:lssmixture}. Therefore the result would follow for any $\epsilon >0$ if we can find a $\tilde P$ inducing a $\tilde{f}$ which makes both summands of the right hand side of \eqref{eq:kl} less than $\epsilon/2$ for every $f$ induced by $P \in \mathcal{W}$, with $\mathcal{W}$ a weak neighborhood of $\tilde P$.

First we show how to construct $\tilde P$. Let 
\[
	\d P_n(\xi, \omega, \lambda) = t_n \I_{\xi \in [-n,n]}f_0(\xi) \delta_{\omega_n}(\omega) \delta_{\lambda_n}(\lambda)
\]
where $\omega_n = n^{-\eta}$, $\lambda_n = 1 + n^{-\beta}$, with $\beta>\eta$, $t_n = (\int_{[-n,n]} f_0(t)\d t)^{-1}$, $\I$ is the indicator function and $\delta_x$ is the Dirac delta mass at a point $x$. Under these assumptions we introduce the density $f_n$ induced by the mixture \eqref{eq:lssmixture} with $P = P_n$, that is
\[
	f_n(x) = t_n \int_{[-n,n]} \frac{2}{\omega_n} \phi\left(\frac{x-\xi}{\omega_n}\right) \Phi\left(\lambda_n \frac{x-\xi}{\omega_n}\right) f_0(\xi) \d \xi.
\]
Now recall that the cdf of $Z \sim SN(0,1,\lambda)$  can be written as the sum of two terms, namely
\begin{equation}
\text{Pr}(Z<z) = \Phi(z) - 2 T(z,\lambda),
\end{equation}
where $T(z,\lambda)$ is the Owen's $T$ function defined as
\begin{equation*}
 T(h,a)=\frac{1}{2\pi}\int_{0}^{a} \frac{e^{-\frac{1}{2} h^2 (1+x^2)}}{1+x^2}  dx.
\end{equation*}
A useful property of the Owen's $T$ is that $T(-h,a) = T(h,a)$. Hence
\begin{align*}
\int_{[-c,c]}2 \phi(t) \Phi(\lambda t) \d t & = \int_{[-\infty,c]}2 \phi(t) \Phi(\lambda t) \d t - \int_{[-\infty,-c]}2 \phi(t) \Phi(\lambda t) \d t \\
	& = \Phi(c) - 2 T(c,\lambda)  - \Phi(-c) + 2 T(-c,\lambda)\\
	& = \int_{[-c,c]} \phi(t) \d t.
\end{align*} 
Thus, find a constant $c$ so that $\int_{[-c,c]}\phi(t) \d t > 1 - \epsilon$ and fix an $x \in \R$. For $n$ large, one obtains
\begin{equation*}
\inf_{y \in (x-c\omega_n, x+c\omega_n)}f_0(y) (1-\epsilon) 
< \frac{f_n(x)}{t_n} <
\sup_{y \in (x-c\omega_n, x+c\omega_n)}f_0(y) + M\epsilon.
\end{equation*} 
Since $t_n \to 1$ and $\omega_n \to 0$, the equation above implies that
\begin{equation}
	\log  \frac{f_0(x)}{f_n(x)} \to 0 \text{ for all $x \in \R$}.
\label{eq:seqforDMC}
\end{equation}
To conclude the first part of the proof we use the dominated convergence theorem to show that $\int f_0(x) \log  \frac{f_0(x)}{f_n(x)} \d x \to 0  $. This is done mainly borrowing the techniques used in \citet{tokd:2006} and \citet{art:wu:ghos:2008} while adapting them to our location-scale-shape mixture setting.  Find a function $g(x)$ such that it is $f_0$ integrable and that dominates the absolute value of the left hand side of \eqref{eq:seqforDMC}. An upper bound for $f_n(x)$ can be obtained as a consequence of assumption ($i$) of Theorem~\ref{theo:support}.
\begin{equation}
	f_n(x) < Mt_n < M t_1.
\label{eq:upboundfn}
\end{equation}

To find a lower bound for $f_n(x)$, first consider $|x| \geq  n$, where we have
\begin{align*}
	f_n(x)  & = t_n \int_{[-n,n]} \frac{2}{\omega_n} \phi\left(\frac{x-\xi}{\omega_n}\right) \Phi\left(\lambda_n \frac{x-\xi}{\omega_n}\right) f_0(\xi) \d \xi \\
		& \geq  t_n \int_{[-n,n]} \frac{2}{\omega_n} \phi\left(\frac{x +n}{\omega_n}\right) \Phi\left(\lambda_n \frac{x+n}{\omega_n}\right) f_0(\xi) \d \xi \\
		& =   \frac{2}{\omega_n} \phi\left(\frac{x +n}{\omega_n}\right) \Phi\left(\lambda_n \frac{x+n}{\omega_n}\right)  \\
		& =   2 n^{\eta} \phi\left( xn^{\eta} + n^{1+\eta} \right) \Phi\left\{(1+n^{-\beta} )( xn^{\eta} + n^{1+\eta}) \right\}  \\
		& \geq   2 n^{\eta} \phi\left( xn^{\eta} + n^{1+\eta} \right) \Phi\left(x \right)  \\
		& \geq   2 |x|^{\eta} \phi\left( 2|x|^{1+\eta} \right) \Phi\left(x\right)
\end{align*}
where the last but one inequality is true since $\Phi$ is monotone increasing, and the last one follows from the fact that $\tau^\eta \phi(\tau^\eta (|x| +\tau))$ is decreasing for $\tau >1$. 
Define $\psi_n(x)$ satisfying assumption $(iii)$ as $\psi_n(x) = \inf_{t\in (x-\omega_n,x+\omega_n)} f_0(t)$,
and let $A_n = [-n,n] \bigcap (x-\omega_n,x+\omega_n)$. Then, for $1 \geq n \geq |x|$, we have
\begin{align*}
	f_n(x)  & \geq t_n \int_{A_n} \frac{2}{\omega_n} \phi\left(\frac{x-\xi}{\omega_n}\right) \Phi\left(\lambda_n \frac{x-\xi}{\omega_n}\right) f_0(\xi) \d \xi \\
		& \geq  t_n \psi_n(x) \int_{A_n} \frac{2}{\omega_n} \phi\left(\frac{x -\xi}{\omega_n}\right) \Phi\left(\lambda_n \frac{x-\xi}{\omega_n}\right) \d \xi \\
		& \geq c\psi_1(x).
\end{align*}
Now consider
\begin{align*}
\int_{A_n} \frac{2}{\omega_n} \phi\left(\frac{x -\xi}{\omega_n}\right) \Phi\left(\lambda_n \frac{x-\xi}{\omega_n}\right) \d \xi 
& \geq \int_{0}^1 2 \phi\left(t \right) \Phi\left(-\lambda_n t\right) \d t \\
& \geq \int_{0}^1 2 \phi\left(t \right) \Phi\left(-t\right) \d t  \\
& \geq \int_{0}^1 \phi\left(t \right) \d t - [\Phi(1) \{1-\Phi(1)\} + 1/4] =  c,
\end{align*}
where the last inequality is true since $2T(h,1) = \Phi(t) \{1-\Phi(t)\}$, by the properties of the Owen's $T$ function. 
Therefore for $0 < R < n$, 
\begin{equation}
	\log \frac{f_0(x)}{f_n(x)} \leq u(x) = 
	\begin{cases}
		\log \frac{f_0(x)}{c \psi_1(x)} & \text{for $|x|<R$} \\
		\max \left\{ \log \frac{f_0(x)}{c \psi_1(x)}, \log \frac{f_0(x)}{  2 |x|^{\eta} \phi\left( 2|x|^{1+\eta} \right) \Phi\left(x \right)   } \right\} & \text{for $|x|\geq R$} \\
	\end{cases}
\label{eq:loboundfn}
\end{equation}

Using \eqref{eq:upboundfn} and \eqref{eq:loboundfn} we have
\[
 \left| \log \frac{f_0(x)}{f_n(x)} \right| \leq \max\left\{ \left|\log \frac{f_0(x)}{M t_1}\right|, u(x)  \right\}.
\]
Since the above quantity is $f_0$ integrable, by assumptions ($ii$), ($iii$), and ($iv$), we can invoke the dominated convergence theorem and conclude $ f_0(x) \log  \frac{f_0(x)}{f_n(x)} \d x \to 0 $. 
Letting $\tilde{f} = f_n$, for large $n$, the first summand of equation \eqref{eq:kl} can be made less than $\epsilon/2$.
Using the following Lemma~\ref{lem:astokdar}, which is similar to Lemma 3.1 of \citet{tokd:2006}, we are able to find an upper bound for the second summand and this concludes the proof.  
\end{proof}

\begin{lemma}
\label{lem:astokdar}

Consider an $f_0$ with finite second moment. Suppose that $\tilde f$ is induced by the mixture prior \eqref{eq:lssmixture} and the mixing distribution $\tilde P$ with $\tilde{P}\{ (-a,a) \times (\omega_l, \omega_u) \times (0,s)\} = 1$ for given $a > 0$, $0 < \omega_l < \omega_u$ and $s>0$. 
Then for any $\epsilon/2 > 0$, there exists a weak neighborhood $\cal W$ of $\tilde P$  such that for any $f$ induced via the mixture \eqref{eq:lssmixture} and $P \in \mathcal{W}$,
\[
	\int f_0(x) \log \frac{\tilde{f}(x)}{f(x)} \d x \leq \epsilon/2.
\]
\end{lemma}

\begin{proof}
The proof follows from \citet{tokd:2006} with minor adaptations. The main difference is in showing that his (A.1) is true under our location-scale-shape mixture prior.
Denote by $G \subset \R \times \R^+ \times \R$ the set $\{(-a,a) \times (\omega_l, \omega_u) \times (-s,s)\}$. Choose $k > a +\omega_u$ such that 
$\int_{|x|>k} (|x|+a)^2/(2\omega_l^2) f_0(x) \d x < \epsilon/4$. Take $V = \{ P : P(G) > \omega_l/\omega_u \}$. Then $V$ is weak neighborhood of $\tilde P$. The following inequality substitutes equation (A.1) of \citet{tokd:2006}:
\begin{align*}
	\int_{|x|>k} f_0 \log \frac{\tilde{f}(x)}{f(x)} \d x & \leq  
	\int_{|x|>k} f_0 \log \frac{\int_G 2\omega^{-1} \phi(\omega^{-1}(x-\xi)) \Phi(\lambda \omega^{-1}(x-\xi) ) \d \tilde{P}(\xi, \omega, \lambda)}{ \int 2\omega^{-1} \phi(\omega^{-1}(x-\xi)) \Phi(\lambda \omega^{-1}(x-\xi) ) \d P(\xi, \omega, \lambda) } \d x \\
	& \leq \int_{|x|>k} f_0 \log \frac{ 2\omega_u^{-1} \phi(\omega_u^{-1}(|x|-a)) }{ 2\omega_l^{-1} \phi(\omega_l^{-1}(|x|+a))
 \Phi(0) P(G) } \d x \\
 	& \leq \int_{|x|>k} f_0(x)\left\{ \frac{(|x|+a)^2}{2\omega_l^2} - \log(4 )\right\} \d x < \epsilon/4.
\end{align*} 
The remaining of the proofs follows that of Theorem 3 of  \citet{ghos:etal:1999} and that of Lemma 3.1 of \citet{tokd:2006}.
\end{proof}

\begin{proof}[Lemma 1]
The prove the result,  we first exploit the stick breaking construction of the Dirichlet process, and considering each single component of the mixture we obtain the entropy of the set $\Theta_{a,l,u,s}=  \{\dSN{y}{\theta}{\omega}{\lambda}: |\theta|\leq a, l < \omega < u, |\lambda| \leq s \}$. This can be done with similar arguments to \citet{tokd:2006} but with particular additional technicalities related to the skew-normal distribution and Owen's $T$ function. Then we follow part of the proof of Theorem 5.10 by \citet{pati:duns:tokd:2011}, in particular, considering $\{\sum_{h=1}^{m} \pi_h f_{SN}(y; \xi_h, \omega_h, \lambda_h): |\xi_h|\leq a, l < \omega_h < u, |\lambda_h| \leq s \text{ for $h=1, \dots, m,$} \sum_{h>m} \pi_h < \epsilon \}$ and obtaining the entropy for the whole sieve \eqref{eq:sieve}. Some passages are similar to those of \citet{pati:duns:tokd:2011} but we write all the proof for sake of completeness. 

For any $f_1$, $f_2 \in \mathcal{F}$ we have 
\begin{equation}
||f_1 - f_2||_1 \leq 
\sum_{l=1}^{m} \pi_l^{(1)} \int |\dSN{y}{\theta_l^{(1)}}{\omega^{(1)}_l}{\lambda^{(1)}_l} -
					   \dSN{y}{\theta_l^{(2)}}{\omega^{(2)}_l}{\lambda^{(2)}_l} | \d y 
	+ \sum_{l=1}^{m} | \pi_l^{(1)} - \pi_l^{(2)}| + 2\epsilon,
\label{eq:sticksieve}
\end{equation}
We start by showing that two single skew-normal kernels with suitable parameters have $L_1$ distance smaller than~$\epsilon$.  To show that, first let $\zeta = \epsilon/6$ and $\eta = 2\tan(\pi\epsilon/16)$. Define $\sigma_m = l(1 + \zeta)^m$, $m \geq 0$. Let $M$ be the smallest integer so that $l(1 + \zeta)^M \geq u$. This clearly implies $M \leq (1 + \zeta)^{-1}\log(u/l) + 1$. For $1 \leq j \leq M$, let $N_j = \lceil \frac{\sqrt{32}}{\sqrt{\pi}} a/(\epsilon/3 \sigma_{j-1}) \rceil$. 
For $1\leq i \leq N_j$; $1\leq j\leq M$, and $1 \leq k \leq \lceil2s/\eta\rceil$ define
\[
	E_{ijk} = \left( -a +  \frac{2a(i-1)}{N_j},  -a +  \frac{2ai}{N_j} \right] \times (\sigma_{j-1}, \sigma_{j}] \times (-s +\eta(k-1), -s + \eta k].
\] 
Assume that $(\theta_1,\omega_1, \lambda_1)$ and $(\theta_2,\omega_2, \lambda_2) \in E_{ijk}$.
Now, to obtain a bound on the $L_1$ distance, first apply the triangular inequality twice, i.e.
\begin{align}
	|| \dSN{y}{\theta_1}{\omega_1}{\lambda_1} -  \dSN{y}{\theta_2}{\omega_2}{\lambda_2} ||_1
	 \leq & 	
	|| \dSN{y}{\theta_1}{\omega_1}{\lambda_1} -  \dSN{y}{\theta_1}{\omega_1}{\lambda_2} ||_1  + \notag \\
	&|| \dSN{y}{\theta_1}{\omega_1}{\lambda_2} -  \dSN{y}{\theta_1}{\omega_2}{\lambda_2} ||_1 + \notag \\
	&|| \dSN{y}{\theta_1}{\omega_2}{\lambda_2} -  \dSN{y}{\theta_2}{\omega_2}{\lambda_2} ||_1 
\label{eq:triangineq}
\end{align}
Then we show that each one of the three summand above is less then $\epsilon/3$. First consider 
\begin{align*}
	|| \dSN{y}{\theta_1}{\omega_1}{\lambda_1} -  \dSN{y}{\theta_1}{\omega_1}{\lambda_2} ||_1  
	& = \int_{-\infty}^\infty | \dSN{y}{\theta_1}{\omega_1}{\lambda_1} -  \dSN{y}{\theta_1}{\omega_1}{\lambda_2} | \d y  \\
 	& = 2 \int_{-\infty}^\infty \phi(y) | \Phi(\lambda_1 y)  -  \Phi(\lambda_2 y ) | \d y  \\
 	& \leq 4 \int_{0}^\infty \phi(y) \{ \Phi(\eta/2 y)  -  \Phi(-\eta/2 y ) \} \d y  \\
 	& = 4 \int_{0}^\infty \phi(y) \{ 1 - 2  \Phi(-\eta/2 y ) \} \d y  \\
 	& = 4\int_{0}^\infty \phi(y)  \d y - 2 \int_{0}^\infty 2 \phi(y)  \Phi(-\eta/2 y )  \d y   \\
 	& = 2 - 4\mbox{Pr}(X_{SN_{-\eta/2}} \geq 0) = 8 T(0,\eta/2) = \frac{4}{\pi} \mbox{atan}(\eta/2),
\end{align*}
where the last equality follows directly from the definition of the Owen's $T$ function.
Clearly since $\eta = 2\tan (\pi \epsilon/16)$, the above summand is bounded by $\epsilon/3$. 
Now consider the second summand of \eqref{eq:triangineq} substituting, without loss of generality $\theta_1=0$ and $\lambda_2= \lambda>0$, and assuming $\omega_2 > \omega_1$, by Csisz\'ar inequality we have
\begin{align*}
	|| \dSN{y}{0}{\omega_1}{\lambda} -  \dSN{y}{0}{\omega_2}{\lambda} ||_1 
 \leq \sqrt{ 2KL  (\dSN{y}{0}{\omega_2}{\lambda},  \dSN{y}{0}{\omega_1}{\lambda})},
\end{align*}
where $KL(f,g)$ is the Kullback-Leibler divergence. Then,
\begin{align*}
	KL  (\dSN{y}{0}{\omega_2}{\lambda},  \dSN{y}{0}{\omega_1}{\lambda}) 
	 = &\int_{-\infty}^\infty  \frac{2}{\omega_2} \phi \left(\frac{y}{\omega_2}\right) \Phi \left(\lambda \frac{y}{\omega_2}\right) 
		\log\left( \frac{\frac{2}{\omega_2} \phi \left(\frac{y}{\omega_2}\right) \Phi \left(\lambda \frac{y}{\omega_2}\right)}{\frac{2}{\omega_1} \phi \left(\frac{y}{\omega_1}\right) \Phi \left(\lambda \frac{y}{\omega_1}\right)} \right) \d y\\
	 = &\int_{-\infty}^\infty  \frac{2}{\omega_2} \phi \left(\frac{y}{\omega_2}\right) \Phi \left(\lambda \frac{y}{\omega_2}\right) 
		\frac{1}{2} y	^2  \left( \frac{\omega_2^2 - \omega_1^2 }{\omega_1^2 \omega_2^2}\right) \d y \, + \\
		& \int_{-\infty}^\infty  \frac{2}{\omega_2} \phi \left(\frac{y}{\omega_2}\right) \Phi \left(\lambda \frac{y}{\omega_2}\right)\log\left( \frac{\omega_1 \Phi \left(\lambda \frac{y}{\omega_2}\right)}{ \omega_2 \Phi \left(\lambda \frac{y}{\omega_1}\right)} \right) \d y\\
	 = & \frac{1}{2} \left( \frac{\omega_2^2 - \omega_1^2 }{\omega_1^2 \omega_2^2}\right) E (y^2) + 
	   E\left( \log\left( \frac{\omega_1 \Phi \left(\lambda \frac{y}{\omega_2}\right)}{ \omega_2 \Phi \left(\lambda \frac{y}{\omega_1}\right)} \right) \right).
\end{align*}
Since the second moment of a skew-normal distribution equals its squared scale parameter and the second expectation above is smaller than the first summand, we have
\begin{align*}
	|| \dSN{y}{\theta_1}{\omega_1}{\lambda_2} -  \dSN{y}{\theta_1}{\omega_2}{\lambda_2} ||_1 
	\leq  2 \frac{\omega_2-\omega_1}{\omega_1} \leq \zeta = \epsilon/3.
\end{align*}
Finally consider the third summand of \eqref{eq:triangineq}. Define $H$ to be an half-normal random variable with density $f_H$, which is obtained when the skew-normal shape parameter $\lambda \to \infty$. Then, 
\begin{align*}
|| \dSN{y}{\theta_1}{\omega_2}{\lambda_2} -  \dSN{y}{\theta_2}{\omega_2}{\lambda_2} ||_1 
	& \leq \int_{-\infty}^\infty \left| f_H\left( \frac{y - \theta_1}{\omega}\right)  -  f_H\left( \frac{y- \theta_2}{\omega}\right) \right| \d y  \\
	& = 2 \text{Pr}\left( H < \frac{\theta_2-\theta_1}{\omega_2}\right) \\
	& = 4 \Phi\left(\frac{\theta_2-\theta_1}{\omega_2}\right) -2  \\
	& =  2 \left(\frac{2}{\pi}\right)^{1/2} \int_0^{\frac{\theta_2-\theta_1}{\omega_2} } \exp\{-y^2/2\} \d y \\
& \leq 2 \left(\frac{2}{\pi}\right)^{1/2} \frac{|\theta_1 - \theta_2|}{\omega_2},
\end{align*}
which is a very loose bound equal to twice the same bound of the $L_1$ distance between two Gaussian distribution with same variance and different locations. See \citet{tokd:2006}, equation (4.1). It follows that also this last summand is less than $\epsilon/3$ which leads to
\[
|| \dSN{y}{\theta_1}{\omega_1}{\lambda_1} -  \dSN{y}{\theta_2}{\omega_2}{\lambda_2} ||_1 \leq \epsilon.
\]
Hence the partition given by $E_{ijk}$ induces a covering of $\Theta_{a,l,u,s}=  \{\dSN{y}{\theta}{\omega}{\lambda}: |\theta|\leq a, l < \omega < u, |\lambda| \leq s \}$ made of $\epsilon$ balls. If $N$ is the minimum number of $\epsilon$ ball we need to cover $\Theta_{a,l,u,s}$, we have

\begin{align}
	N  & \leq \left( \frac{2s}{\eta} + 1\right) \sum_{j = 1}^M  \left( \sqrt{\frac{32}{\pi}} \frac{a}{\epsilon/6 \sigma_{j-1}} +1 \right)  \notag \\
		& \leq  \left( c_1 s + 1\right) \left( c_2 \frac{a}{l} + c_3 s \log\frac{u}{l} +  1\right) \notag\\
		& \leq   d_1 \left(\frac{as}{l}\right) + d_2 \left(\frac{a}{l}\right) + d_3 s \log\left(\frac{u}{l}\right) + d_4 \log\left(\frac{u}{l}\right)  + s +1,
\label{entropy1}
\end{align}
where $c_1, c_2, c_3, d_1, d_2, d_3$, and $d_4$ are positive constants depending on $\epsilon$.

The entropy in \eqref{entropy1} is determinant in computing the entropy of the whole sieve. Going back to the stick-breaking representation \eqref{eq:sticksieve}, let $\Theta_{\pi} = \{ \pi^{m} = (\pi_1, \dots, \pi_m)\}$. Fix $\pi_1^{m}$ and $\pi_2^{m} \in \Theta_{\pi}$.
Let for $k = 1,2$, $V_h^{(k)} = \pi_h^{(k)}(1-\sum_{l<h}\pi_l^{(k)})$. Clearly $ \sum_{h=1}^{m} | \pi_h^{(1)} - \pi_h^{(2)}| < \epsilon$ if for each $h=1,\dots,m$, $|V_h^{(1)} - V_h^{(2)}| < \epsilon/m^2$. 
Since $V_h^{(1)}$, $V_h^{(2)} \in [0,1]$, the number of $\epsilon$-balls required to cover $\Theta_{\pi}$ is $(m^2/\epsilon)^{m}$ times a constant. Hence
\[
J(\epsilon,\mathcal{F}_{a,u,l,s,m}) \leq m \log\left\{ d_1 \left(\frac{as}{l}\right) + d_2 \left(\frac{a}{l}\right) + d_3 s \log\left(\frac{u}{l}\right) + d_4 \log\left(\frac{u}{l}\right)  + s +1 \right\} + d_3 m \log(d_4 m).
\]
\end{proof}

\bibliographystyle{apalike}
\bibliography{biblio}

\end{document}